\documentclass[conference]{IEEEtran}
\IEEEoverridecommandlockouts
\usepackage{cite}
\usepackage{amsmath,amssymb,amsfonts}
\usepackage{algorithmic}
\usepackage{algorithm}
\usepackage{graphicx}
\usepackage{textcomp}
\usepackage{xcolor}
\usepackage{caption}
\usepackage{multicol}
\usepackage{multirow}
\usepackage{pifont}
\usepackage{bbding}
\usepackage{url}

\renewcommand\baselinestretch{0.946}

\def\BibTeX{{\rm B\kern-.05em{\sc i\kern-.025em b}\kern-.08em
    T\kern-.1667em\lower.7ex\hbox{E}\kern-.125emX}}
\begin{document}

\title{Mera: Memory Reduction and Acceleration for Quantum Circuit Simulation via Redundancy Exploration \\
\thanks{This work was supported by the National Natural Science Foundation of China (Grant No. 62372182, 62372183).}
}

\author{
\IEEEauthorblockN{
Yuhong Song$^{1, 2}$ \quad
Edwin Hsing-Mean Sha$^{1}$ \quad
Longshan Xu$^{1}$ \quad
Qingfeng Zhuge$^{1,* \thanks{* Qingfeng Zhuge is the corresponding author}}$ \quad
Zili Shao$^{2}$
}
\\
\IEEEauthorblockA{
\normalsize
$^{1}$ School of Computer Science and Technology, East China Normal University, Shanghai 200062, China \\
$^{2}$ Department of Computer Science and Engineering, The Chinese University of Hong Kong, Hong Kong, China \quad
}
\IEEEauthorblockA{
\normalsize
Email: yhsong@stu.ecnu.edu.cn, edwinsha@cs.ecnu.edu.cn, lsxu@stu.ecnu.edu.cn, \\qfzhuge@cs.ecnu.edu.cn, shao@cse.cuhk.edu.hk}
}


\maketitle

\begin{abstract}
With the development of quantum computing, quantum processor demonstrates the potential supremacy in specific applications, such as Grover's database search and popular quantum neural networks (QNNs). 
For better calibrating the quantum algorithms and machines, quantum circuit simulation on classical computers becomes crucial.
However, as the number of quantum bits (qubits) increases, the memory requirement grows exponentially.
In order to reduce memory usage and accelerate simulation, we propose a multi-level optimization, namely Mera, by exploring memory and computation redundancy.
First, for a large number of sparse quantum gates, we propose two compressed structures for low-level full-state simulation.
The corresponding gate operations are designed for practical implementations, which are relieved from the long-time compression and decompression.
Second, for the dense Hadamard gate, which is definitely used to construct the superposition, we design a customized structure for significant memory saving as a regularity-oriented simulation.
Meanwhile, an on-demand amplitude updating process is optimized for execution acceleration.
Experiments show that our compressed structures increase the number of qubits from 17 to 35, and achieve up to 6.9$\times$ acceleration for QNN.
\end{abstract}

\begin{IEEEkeywords}
Quantum circuit simulation, memory reduction, acceleration, redundancy exploration, algorithm optimization
\end{IEEEkeywords}

\section{Introduction}
Noisy intermediate-scale quantum (NISQ) machines have been proved their quantum supremacy \cite{boixo2018characterizing}.
Benefiting from the superposition and entanglement of quantum bit (qubit), quantum computers demonstrate huge advantages over classical computers in some specific applications, like Shor's integer factorization \cite{shor1999polynomial}, Grover's database search \cite{grover1996fast}, and quantum neural networks (QNNs) \cite{jiang2021co,hu2022quantum}. 
Because of the immaturity and huge overhead of existing physical quantum processors, simulation on classical computers becomes crucial for better understanding quantum behaviors.
Among different types of quantum simulators \cite{simulator_list}, full-state simulation \cite{de2007massively} becomes important, because it allows deeper and larger quantum circuit simulation.
It updates the state-vector in each time step.

However, as the number of simulated qubits increases, classical computers cannot afford the huge memory requirements.
Because the full-state quantum circuit simulation relies on matrix tensor-product (MTP) operation, which produces exponential memory increment on operation-matrix and state-vector.
An $n$-qubit quantum circuit produces an operation-matrix and a state-vector with the size of $2^n \times 2^n$ and $2^n$ in each time step.
Each element is represented using double-precision complex numbers. 
As a result, the operation-matrix requires $2^{2n+4}$ bytes and the state-vector needs $2^{n+4}$ bytes to store.
As the number of qubits increases, the memory undoubtedly becomes the bottleneck of quantum circuit simulation.

Based on this, there are lots of works \cite{haner20175,li2019quantum} that rely on huge-memory supercomputers to mitigate the memory bottleneck.
However, even supercomputers still exist memory caps that can’t keep up with exponential memory growth. 
Without any optimization, the state-of-the-art Fugaku supercomputer, which possesses 4.9 PB memory\cite{top500}, can only accommodate 24 qubits.
Meanwhile, the hardware and software platforms of supercomputers are not easy to access, generous software debugging for quantum algorithms design will lead to heavy overhead.
Other works \cite{smelyanskiy2016qhipster,khammassi2017qx} focus on compilation-level optimizations and parallel computing in distributed systems to mitigate the memory problem and accelerate computation. 
But there is still plenty of memory and computation waste that can be explored and removed.
In this work, we focus on redundancy exploration and algorithm optimizations on modern workstations to reduce memory usage and accelerate quantum circuit simulation.

First, because the quantum circuit simulation depends on the gate MTP operation, we explore forty commonly used quantum gates.
There is a great percentage of sparse gates with duplicated zero values, which bring huge potential for memory saving.
Aiming at the sparse quantum gates, we design two new compressed structures, namely DAX and DAS, for low-level full-state simulation.
We encode the gate from scratch and design corresponding MTP and matrix-vector multiplication (MVM) calculations, so it doesn't require long-time compression and decompression processes.
However, the Hadamard gate, which is definitely used to construct the quantum superposition, is a dense matrix without any zero elements that can be compressed.
Therefore, we especially explore the regular patterns of the Hadamard gate and propose a novel customized structure for a regularity-oriented simulation.
Meanwhile, the corresponding MVM operation for the next state-vector acquisition is optimized for efficient simulation.

The main contributions of this paper are listed as follows.
\begin{itemize}
    \item First, we explore the sparsity of quantum gates and propose two new compressed structures to encode the quantum gates for a low-level simulation. 
    The corresponding MTP and MVM operations are designed for acceleration by abandoning the long-time compression and decompression.
    \item For the dense Hadamard gate, we explore the computation redundancies and design a novel customized structure for regularity-oriented simulation.
    The MVM computation is optimized by parsing the structure on demand.
    \item Experiments show that the compressed structures can increase the number of qubits from 17 to 34.
    Moreover, it achieves up to 1800.3$\times$ and 141.2$\times$ speedup for MTP and MVM operations.
    For dense Hardamard gates, our method achieves up to 6.9$\times$ and 4.6$\times$ simulation acceleration for QNN and Grover's algorithms.
\end{itemize}


\section{Background, Related Works and Motivations} \label{sec:BRM}

\subsection{Basics of quantum circuit simulation}
Given a quantum circuit, each qubit is a two-level quantum system.
Its state $|\psi\rangle$ can be expressed by two complex number $\alpha_0$ and $\alpha_1$ as
\begin{equation}
\label{eqa:qubit}
    |\psi\rangle = \alpha_0 |0\rangle + \alpha_1 |1\rangle,
\end{equation}
where $|\alpha_0|^2 + |\alpha_1|^2 = 1$.
$|0\rangle$ and $|1\rangle$ are two computational orthonormal basic states.
The one-qubit state $|\psi\rangle$ can also be represented using two orthonormal vectors as
\begin{equation}
\begin{small}
|\psi\rangle = \alpha_0
    \left[
    \begin{matrix}
        1 \\
        0 \\
    \end{matrix}
    \right]
  + \alpha_1 
    \left[
    \begin{matrix}
        0 \\
        1 \\
    \end{matrix}
    \right]
  = \left[
    \begin{matrix}
        \alpha_0 \\
        \alpha_1 \\
  \end{matrix}
  \right],
\end{small}
\end{equation}
where {\footnotesize $\left[
    \begin{matrix}
        1 \\
        0 \\
    \end{matrix}
    \right]$}
  and {\footnotesize $\left[
        \begin{matrix}
        0 \\
        1 \\
        \end{matrix}
        \right]$}
are the vector representation of $|0\rangle$ and $|1\rangle$, respectively.
This vector is named as state-vector.
The representation means that the qubit is in a superposition of {\small $|0\rangle$} and {\small $|1\rangle$}.
The square of each element (i.e., amplitude) refers to the probability of the qubit collapsing into the specific state.
When the qubit is measured, it has $p = |\alpha_0|^2$ probability of collapsing to state $|0\rangle$ and $1 - p = |\alpha_1|^2$ probability of collapsing to state $|1\rangle$.

Given an $n$-qubit quantum circuit, there are $2^n$ amplitudes in the state-vector.
More generally, the state-vector can be represented as 
\begin{equation} \label{equa:sum1}
\begin{aligned}
    |\phi\rangle = 
    \alpha_{00..0} |00..0\rangle +
    ... +
    \alpha_{11..1} |11..1\rangle =
    \left[
        \begin{matrix}
            \alpha_{00..0} \\
            \alpha_{00..1} \\
            ... \\
            \alpha_{11..1}
        \end{matrix}
    \right].
\end{aligned}
\end{equation}
The squared amplitudes have to sum up to 1, that is
\begin{equation}
    \sum_{i=0}^{N-1}|\alpha_i|^2 = 1,
\end{equation}
where the $N = 2^n$.
When the $n$-qubit quantum circuit is measured, it has $|\alpha_{00..0}|^2$ probability of collapsing to state $|00..0\rangle$, etc. 

In quantum circuit simulation, quantum gates are applied to implement the circuit transformation.
Any quantum gate is represented as a complex unitary matrix.
Given an $n$-qubit quantum circuit, if we apply a single-qubit gate $G$ into the $k$th qubit, the transformation can be represented as a $2^n \times 2^n$ unitary matrix as 
\begin{equation}
    U = I^{\otimes k-1} \otimes G \otimes I^{\otimes n-k},
\end{equation}
where $\otimes$ refers to the matrix tensor-product operation, $I$ represents the $2 \times 2$ identity matrix, and $G$ identifies a $2 \times 2$ unitary gate matrix. 
After adding the gate $G$, the state-vector $|\phi\rangle$ can be transformed to
\begin{equation}
    U|\phi\rangle = U \left[
        \begin{matrix}
            \alpha_{00..0} \\
            \alpha_{00..1} \\
            ... \\
            \alpha_{11..1}
        \end{matrix}
    \right],
\end{equation}
which is conducted using the matrix-vector multiplication operation to obtain the state-vector after $U$.

\subsection{Related Works} 
\textbf{quantum circuit simulation.}
As the quantum simulators develop, there are different types of methods \cite{bernstein1997quantum,boixo2017simulation, de2007massively} for time-space trade-off. 
Feynman path simulation \cite{bernstein1997quantum} calculates all the amplitudes by following all the paths from the final states to the initial states.
Both the memory usage and simulation time grow exponentially as the circuit depth increases.
Tensor network contraction \cite{boixo2017simulation,pednault2017breaking} utilizes tensor network to represent quantum circuits \cite{mccaskey2018validating}. 
However, its space overhead and execution time are also exponential as the tree-width of underlying graphs increases.
Based on this, some approaches \cite{chen2018classical,chen201864} compute only one amplitude or a part of the state-vector to mitigate the memory bottleneck.
However, it can't effectively support the intermediate measurement \cite{barends2014superconducting} and full-state assertion checking \cite{huang2019statistical} for software debugging.
For the full-state simulation, the memory requirement grows exponentially as the number of qubits increases.
But the simulation time is polynomial as the circuit depth goes deeper, which enables deeper and larger quantum circuits.
Meanwhile, it can suitably support software debugging.

However, the memory boundary is still the bottleneck for full-state simulation.
Lots of works rely on supercomputers to mitigate the memory problem \cite{haner20175,li2019quantum}, but the hardware and software platforms of supercomputers are not easily accessed by anyone.
Others \cite{smelyanskiy2016qhipster, khammassi2017qx} rely on complication-level optimizations and parallel computing in distributed systems to solve memory problems and accelerate simulation, but there are still huge memory and computation redundancies.
In this paper, we focus on the high-level data structure and algorithms optimizations on personal computers to reduce data redundancies and accelerate the simulation.

\begin{table*}[t]
\captionsetup{font={small}}
\setlength{\abovecaptionskip}{3pt}
\caption{Conclusion of forty commonly used quantum gates. The `Shape' refers to the shape of the gate matrix. `Ratio' identifies the zero ratio of a gate matrix. `Sparsity (n)' refers to the sparsity for identical gate MTP in an $n$-qubit circuit. `Mem Imprv.' means the improvement of memory saving using our compressed structure.}
\def\arraystretch{1.0}\tabcolsep 2.5pt
\def\thefootnote{a}\footnotesize
\begin{center}
\begin{tabular}{|c|c|c|c|c||c|c|c|c|c|}
\hline
\textbf{Number} & \textbf{Gate} & \textbf{Shape/Ratio} & \textbf{Sparsity (n)} & \textbf{Mem imprv.} & \textbf{Number} & \textbf{Gate} & \textbf{Shape/Ratio} & \textbf{Sparsity (n)} & \textbf{Mem imprv.} \\
\hline
\hline
1     & Hadamard & \multicolumn{1}{c|}{\multirow{8}[4]{*}{[2, 2] / 0\%}} & \multicolumn{1}{c|}{\multirow{8}[4]{*}{0}} & \multirow{8}[4]{*}{-} & 21    & RXX   & \multicolumn{1}{c|}{\multirow{4}[2]{*}{[4, 4] / 50\%}} & \multicolumn{1}{c|}{\multirow{4}[2]{*}{$1 - \frac{1}{1.4^n }$}} & \multicolumn{1}{c|}{\multirow{4}[2]{*}{$\frac{2}{3}*2^{\frac{n}{2}} \times$}} \\
2     & square root of Pauli-X &       &       &       & 22    & RYY   &       &       &  \\
3     & SXdg  &       &       &       & 23    & RZX   &       &       &  \\
4     & R     &       &       &       & 24    & ECR   &       &       &  \\
\cline{6-10}5     & RX    &       &       &       & 25    & CH    & \multicolumn{1}{c|}{\multirow{6}[4]{*}{[4, 4] / 62.5\%}} & \multicolumn{1}{c|}{\multirow{6}[4]{*}{$1 - \frac{1}{1.6^n} $}} & \multicolumn{1}{c|}{\multirow{6}[4]{*}{$\frac{2}{3}*\frac{8}{3}^{\frac{n}{2}} \times$}} \\
6     & RY    &       &       &       & 26    & CSX   &       &       &  \\
7     & U     &       &       &       & 27    & CH    &       &       &  \\
8     & U2    &       &       &       & 28    & CSX   &       &       &  \\
\cline{1-5}9     & Identity & \multicolumn{1}{c|}{\multirow{10}[4]{*}{[2, 2] / 50.00\%}} & \multicolumn{1}{c|}{\multirow{10}[4]{*}{$1 - \frac{1}{2^n}$}} & \multicolumn{1}{c||}{\multirow{10}[4]{*}{$\frac{2}{3}*2^n \times$}} & 29    & CU    &       &       &  \\
10    & Pauli-X(NOT) &       &       &       & 30    & CU3   &       &       &  \\
\cline{6-10}11    & Pauli-Y &       &       &       & 31    & RZZ   & \multicolumn{1}{c|}{\multirow{10}[4]{*}{[4, 4] / 75\%}} & \multicolumn{1}{c|}{\multirow{10}[4]{*}{$1 - \frac{1}{2^n}$}} & \multicolumn{1}{c|}{\multirow{10}[4]{*}{$\frac{2}{3}*2^n \times$}} \\
12    & Pauli-Z &       &       &       & 32    & Swap  &       &       &  \\
13    & S     &       &       &       & 33    & iSwap  &       &       &  \\
14    & T     &       &       &       & 34    & Controlled-X &       &       &  \\
15    & T-adjoint  &       &       &       & 35    & Controlled-Y &       &       &  \\
16    & Phase  &       &       &       & 36    & Controlled-Z &       &       &  \\
17    & RZ    &       &       &       & 37    & DCX   &       &       &  \\
18    & U1    &       &       &       & 38    & CPhase  &       &       &  \\
\cline{1-5}19    & CCX (Toffoli) & \multicolumn{1}{c|}{\multirow{2}[2]{*}{[8, 8] / 87.5\%}} & \multicolumn{1}{c|}{\multirow{2}[2]{*}{$1 - \frac{1}{2^n} $}} & \multicolumn{1}{c||}{\multirow{2}[2]{*}{$\frac{2}{3}*2^n \times$}} & 39    & CRZ   &       &       &  \\
20    & CSwap  &       &       &       & 40    & CU1   &       &       &  \\
\hline
\end{tabular}%
\end{center}
\label{tab:gate_spar}%
\end{table*}


\textbf{Data Compression.}
In order to accommodate more qubits on classical computers, data compression methods are commonly used to save memory usage \cite{wu2019full}.
For data compression techniques, works mainly adopt two types of compressors: lossless compressors and error-bounded lossy compressors.
The lossless methods \cite{huffman1952method,witten1987arithmetic} maintain complete original information, so the compressed data can be decoded without any loss.
Relatively, the lossless methods possess low compression ratio.
On the contrary, the lossy compressors \cite{liang2018error,sasaki2015exploration} achieve high compression ratio but lose original information, which makes the inaccuracy problem in quantum circuit simulation.
Whether the lossless compression or lossy compression methods, they all store the compressed data in advance.
During the calculation, the compressed data are first decompressed.
After computation, the results are re-compressed and stored in memory.
The compression and decompression time is long relative to the overall computation time.
Although memory usage is reduced, simulation time is extended.
In this paper, we design data structure and corresponding algorithms to accelerate the quantum circuit simulation without any compression and decompression time.

\subsection{Motivations}\label{sec:challenges}

\textbf{The sparsity of a huge fraction of quantum gates offers the optimization potential.}
Because the full-state quantum circuit simulation relies on the MTP and MVM of gate matrices and state-vector, we first explore the commonly used quantum gates and conclude forty gate matrices.
We observe that there are lots of zero elements in the gate matrices.
The sparsity in the quantum gates is ubiquitous. 
We demonstrate the forty gates, including single-qubit gates and multiple-qubit gates, and their sparsity ratio in Table \ref{tab:gate_spar}.
According to our statistics, 80\% of quantum gates possess a zero ratio over 50\%, and there are 30\% of gates possessing over 75\% zero elements, which makes huge memory waste.
The huge sparsity of quantum gates provides optimization potentials.

Meanwhile, we find that, for a quantum circuit with random gates applying, the sparsity of the operation-matrix after MTP will not be lower than the maximum sparsity ratio among the used gates.
The theorem is proved in the Theorem \ref{them:decrease}. 
That's to say, the results of MTP are highly sparse the same as the used gates.
In order to provide the Theorem \ref{them:decrease}, the Theorem \ref{them:sparsity_ratio} needs to be brought up first.
\newtheorem{theorem}{\textbf{Theorem}}
\begin{theorem}\label{them:sparsity_ratio}
Let $R(A)$ and $R(B)$ be the zero ratios of matrix $A$ and $B$.
Thus, the sparsity of the $A$ and $B$ MTP result is
\begin{equation}\label{equa:AtpB}
    R(A\otimes B) = R(A) + (1 - R(A)) \cdot R(B).
\end{equation}
\end{theorem}
As the MTP proceeds, two elements $a$ and $b$, which are from $A$ and $B$ respectively, conduct multiplication for the corresponding location in the result.
If either $a$ or $b$ is zero, then the corresponding element becomes zero in $A \otimes B$.
Therefore, the sparsity ratio of $A \otimes B$ includes two cases: element $a$ is zero, or element $b$ is zero.

\newtheorem{corollary}{\textbf{Corollary}}
\begin{corollary}\label{col:AtpB}
Let $R(A)$ and $R(B)$ be the zero ratio of matrix $A$ and $B$, we have
\begin{equation}\label{equa:AtpBtpA}
    R(A\otimes B) = R(B\otimes A).
\end{equation}
\end{corollary}
Similarly, according to the Theorem \ref{them:sparsity_ratio}, we can compute the $R(B\otimes A)$, which is equal to $R(A\otimes B)$.

\begin{theorem} \label{them:decrease}
    Let $R(A)$ and $R(B)$ be the zero ratios of matrix $A$ and $B$.
The sparsity ratio of the MTP of $A$ and $B$ $R(A \otimes B) \geq max(R(A), R(B))$.
\end{theorem}
\newtheorem{proof}{\textbf{Proof}}
\begin{proof}
Because $0 \leq R(A), R(B) \leq 1$, we have $1 - R(A) \geq 0$.
Then we have $(1 - R(A)) * R(B) \geq 0$. 
Therefore, $R(A) + (1 - R(A)) * R(B) \geq R(A)$ is established. 
That's, $R(A \otimes B) \geq R(A)$. 
Similarly, $R(B \otimes A) \geq R(B)$ is proved.
According to Corollary \ref{col:AtpB}, we know $R(A \otimes B) = R(B \otimes A)$.
Therefore, $R(A \otimes B) \geq max(R(A), R(B))$ is proved.
\end{proof}

From Theorem \ref{them:decrease}, we conclude that the sparsity ratio of the result can't decrease as the MTP calculation proceeds.
For more sparse gates applying, MTP will produce higher sparsity.

\newtheorem{property}{\textbf{Property}}
\begin{property}\label{pro:AtpB}
Let $R(A)$ be the zero ratio of matrix $A$, the sparsity ratio of the MTP of $m$ identical $A$
\begin{equation}\label{equa:Atpm}
    R(A^{\otimes m}) = 1-(1-R(A))^m
\end{equation}
\end{property}
Above equation is generalized based on the Theorem \ref{them:sparsity_ratio}. The sparsity of $m$ identical quantum gates listed in Table \ref{tab:gate_spar} is shown as column `Sparsity (n)', where the $n = m$ for single-qubit gates, $n = 2m$ for two-qubit gates, $n = 3m$ for three-qubit gates.
For the sparse quantum gates, there are up to $2^n \times$ sparsity for an $n$-qubit circuit. 
As $n$ increases, the sparsity increases significantly.


\section{Design}\label{sec:design}
In order to accommodate more qubits and accelerate quantum circuit simulation on computers with personal computing capability, we propose multi-level optimization methods.
First, for a large part of sparse quantum gates, we propose two compressed structures, namely DAX and DAS, as a low-level simulation.
And the corresponding MTP and MVM operations are designed to accelerate the simulation (see Section \ref{sec:lls}).
Second, for the dense and definitely used Hadamard gate, we propose a special structure, namely RH, based on the computation regularity.
Corresponding MVM operation is optimized for computation speedup (see Section \ref{sec:ros}).

\subsection{Low-level Simulation}\label{sec:lls}

From Section \ref{sec:challenges}, we find that as long as the quantum circuit utilizes gates with high sparsity, the entire operation-matrix of the circuit will also be highly sparse.
Therefore, in this section, we manage to design compressed structures to represent the gate matrix to remove the duplicated zeros.
Moreover, we design the corresponding MTP operation using the compressed structures, which finally produce a compressed operation-matrix.
Our design is relieved from the long compression and decompression. 
Next, we introduce the compressed structures and corresponding computation.

\subsubsection{Compressed structures: DAX and DAS}

In this section, we introduce two structures, namely DAX and DAS,  for data compression.
Except for the non-zero elements, the extra index and distance between non-zero elements are respectively stored to losslessly represent a matrix.
Our DAX and DAS can conduct the MTP among the same structures, so they don't demand long-time compression and decompression.

DAX represents a matrix using non-zero elements and corresponding decimal index.
The complex non-zero data includes real and imaginary parts, which are represented using double-precision.
And the index in the matrix utilizes 64-bit integer to record.
A set of real part, imaginary part, and the index is marked as an entry.
We utilize Y $\otimes$ X transformation as an example to demonstrate the effect of our compressed structure in Figure \ref{fig:dax_and_das} (a).
The MTP of Pauli-Y and Pauli-X gates are shown as follows:
\begin{equation}
\begin{small}
    \nonumber
    Y \otimes X = 
    \left[
    \begin{matrix}
    0 & -i \\
    i & 0
    \end{matrix}
    \right] \otimes
    \left[
    \begin{matrix}
    0 & 1 \\
    1 & 0
    \end{matrix}
    \right] = 
    \left[
    \begin{matrix}
    0 & 0 & 0 & -i \\
    0 & 0 & -i & 0 \\ 
    0 & i & 0 & 0 \\
    i & 0 & 0 & 0
    \end{matrix}
    \right].
\end{small}
\end{equation}
In the $Y \otimes X$ matrix, the non-zero elements are at index 3, 6, 9 and 12 respectively.
Without optimization, the result of $Y \otimes X$ needs 16 * 16B = 256B to store the matrix.
In our DAX structure, we only need 4 * 24B = 96B, which brings 2.7$\times$ memory saving for $Y \otimes X$.

\begin{figure}[t]
\center
\includegraphics[width=0.85\linewidth]{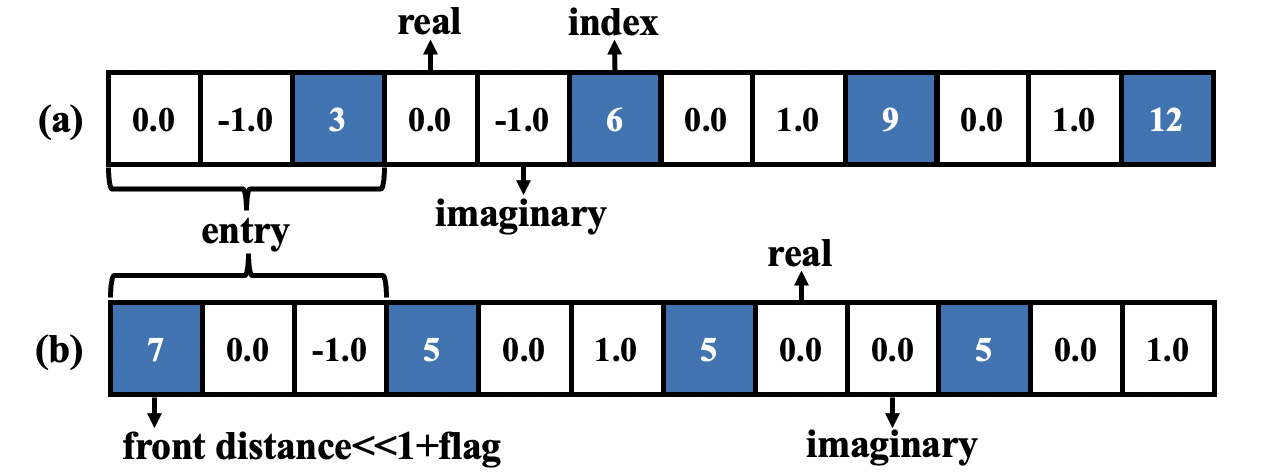}
\captionsetup{font={small}}
\setlength{\abovecaptionskip}{3pt}
\caption{Two compressed structures: (a) DAX and (b) DAS. Utilize Pauli-Y $\otimes$ Pauli-X gates as an example.}
\label{fig:dax_and_das}
\end{figure}

DAS transforms the matrix using non-zero elements and the distance (i.e., the number of zero elements) between two non-zero elements.
The structure is shown in Figure \ref{fig:dax_and_das} (b).
Similar to DAX, an entry of DAS structure includes the distance information, real and imaginary parts of the non-zero elements.
The distance information utilizes a 64-bit integer, 
where the first 63 bits are used to record the number of zeros in front of the non-zero elements and a flag bit is used to mark whether the current non-zero element is the last non-zero element of the matrix row.
Therefore, in the example of $Y \otimes X$, the matrix also needs 4 * 24B = 96B.
It's worth noting that when the matrix size increases, DAX and DAS will obtain more benefits.
We list the memory improvement of forty gates using DAX and DAS structure in Table \ref{tab:gate_spar} (column `Mem imprv.').
The improvement is calculated based on the MTP computation of $m$ identical gates for an $n$-qubit quantum circuit.
DAX and DAS can achieve up to $\frac{2}{3}*2^n \times$ memory saving.
The memory saving becomes larger and larger as $n$ increases.

\subsubsection{Corresponding MTP and MVM design of DAX and DAS}

In this section, we introduce the computation process of MTP and MVM using DAX and DAS structures.
For MTP computation, the original gate matrices are encoded using DAX or DAS.
Then corresponding information can be utilized to participate in computing. 
For later MVM computation, the MTP result can be seamlessly manipulated with current state-vector to produce the updated amplitudes.

For MTP computation, the results are also represented using the compressed structures without decompression and compression.
For the result, we need to compute for the number of non-zero elements, the value and position information of a non-zero element.
First, according to the MTP computation, the number of non-zero elements in results is the product of the number of non-zero elements in two inputs. 
For each non-zero element, the value is computed by the multiplication of two non-zero elements from the inputs.
Last but not least, the index or the distance information for DAX and DAS structures is determined by corresponding information of inputs.


\begin{figure}[t]
\center
\includegraphics[width=0.9\linewidth]{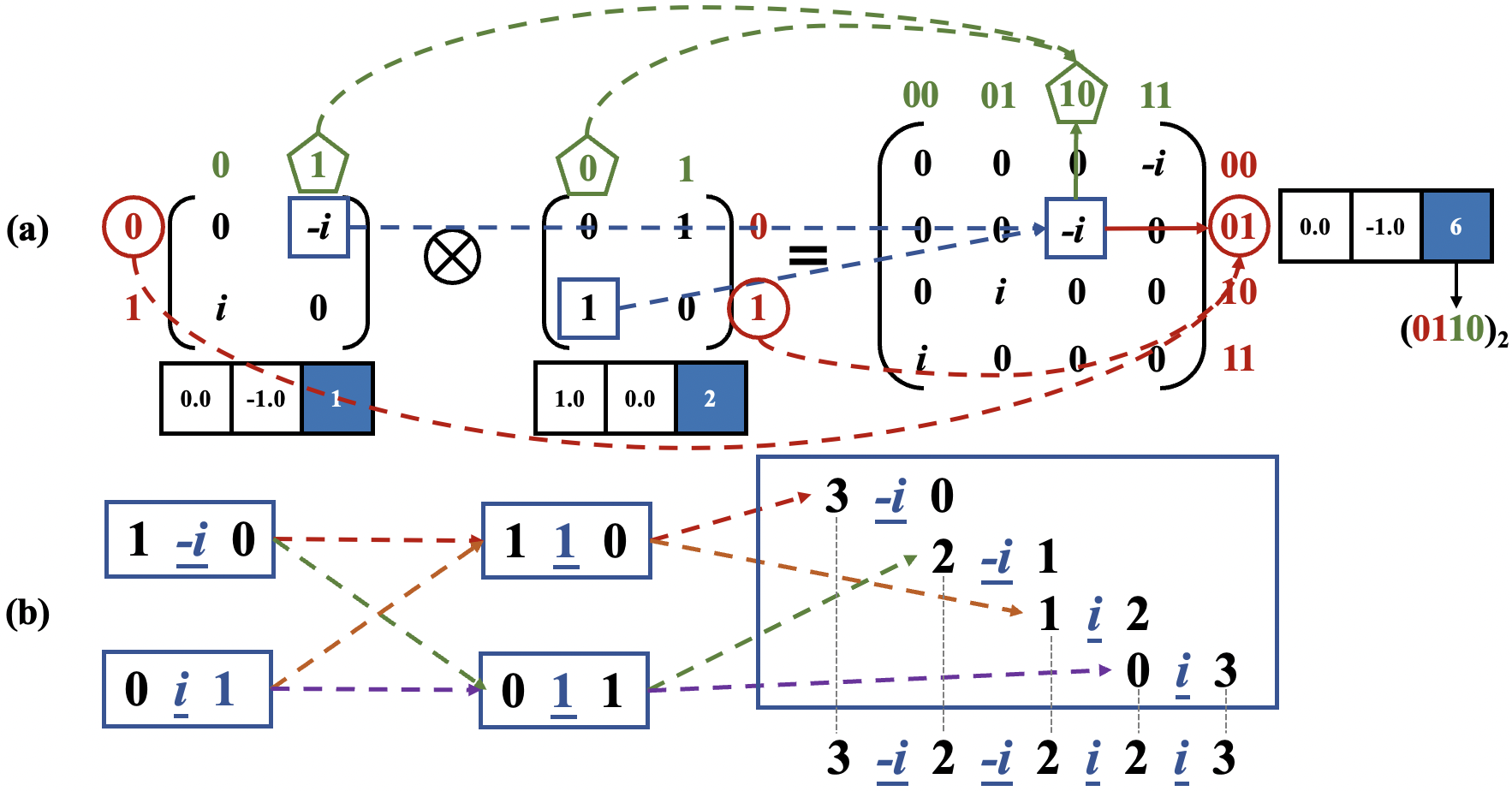}
\captionsetup{font={small}}
\setlength{\abovecaptionskip}{3pt}
\caption{MTP computation of (a) DAX and (b) DAS structures. Utilize Pauli-Y $\otimes$ Pauli-X as an example.}
\label{fig:mtp_dax_das}
\end{figure}

For the DAX structure, the row index and column index of the result are the binary concatenation of the row index and column index of two inputs.
The MTP computation of $Y \otimes X$ using DAX is demonstrated in Figure \ref{fig:mtp_dax_das} (a).
When executing $Y \otimes X$, multiplying the element in the $(r_y)_2$th row and $(c_y)_2$th column of $Y$ by the element in the $(r_x)_2$th row and $(c_x)_2$th column of $X$ equals the element in the $(r_yr_x)_2$th row and $(c_yc_x)_2$th column in result.
For example, in Figure \ref{fig:mtp_dax_das} (a), the non-zero element $-i$ represented by DAX $Y$ is in index $(1)_{10}$, which can be parsed to binary row index $(0)_2$ and column index $(1)_2$
Similarly, the non-zero element $1$ is located in the $(1)_2$th row and $(0)_2$th column of matrix $X$.
Then an entry of the result is computed.
The real and imaginary parts are computed based on the product of two real and imaginary parts of inputs, the results are 0.0 and -1.0 in this example.
And the index is equal to the binary concatenation of the row and column indices, the row index is $(01)_2$ and the column index is $(10)_2$.
Therefore, the decimal index of the corresponding result is $(0110)_2 = (6)_{10}$, which is filled to the index part.
Other non-zero elements are computed by iteratively accessing every non-zero element in the compressed inputs.

For DAS structure, we first parse the distance information of each row, which can obtain the exact number of zeros before and behind the non-zero element within the row.
For example, in Figure \ref{fig:mtp_dax_das} (b), the non-zero element $-i$ has one zero before it and no zero behind it.
Meanwhile, the $-i$ is the last non-zero element within the row.
Based on these distances, the distance of the result in each row can be calculated.
Then we sum up the resulting distances end to end to obtain the accurate compressed result.
For tensor-product of two rows, there are three cases to compute the distance $dis$ before each non-zero element.
\begin{equation}
\begin{small}
    dis = \left\{
    \renewcommand\arraystretch{1.2}
    \begin{array}{ll}
    A_i.dis * B.col + B_i.dis & {case1}\\
    A_i.dis * B.col + B_i.dis + B\_last\_dis  & {case2}\\
    B_i.dis & {otherwise}\\
    \end{array} \right.
\end{small}
\end{equation}
\emph{case1:} $A_i$ and $B_i$ are both the first non-zero elements of the row $A$ and $B$.
\emph{case2:} $B_i$ is the first non-zero element of the row $B$ but $A_i$ is not the first non-zero element of row $A$.
Otherwise, the distance $dis = B_i.dis$.
In above formulas, the $A_i.dis$ is the distance before the $i$th non-zero element in row $A$, $B.col$ means the number of column of row $B$.
Besides, the $B\_last\_dis$ refers to the distance behind the last non-zero element within the row.

For the MVM, we seamlessly implement the multiplication of compressed DAX/DAS structures and the state-vector according to the index and distance information of non-zero elements.
For a given operation-matrix $M$ and state-vector $S$, $S_{next}[i]\ += M[i][j] * S[j]$, where $i$ and $j$ are the row and column indices of $M$ and $S$.
For the MVM result, the $i$ and $j$ need to be iteratively accessed once.
For the matrix $M$, only the non-zero elements have effects on the result $S_{next}$.
Therefore, we only need to access each non-zero element, which is represented using DAX or DAS, then parse the row index $r$ and column index $c$.
Afterwards, the $M[r][c]$ and $S[c]$ are multiplied, corresponding result is added to $S_{next}[r]$.

\subsection{Regularity-oriented Simulation} \label{sec:ros}

From Table \ref{tab:gate_spar}, we find that our DAX and DAS can obtain high compression ratio for sparse quantum gates.
The execution time will be significantly reduced, as shown in Section \ref{sec:exp_dax_das}.
However, there are still some dense gates like Hadamard gate, which is definitely used to construct the quantum superposition by adding $n$ Hadamard gates to the circuit.
Our proposed DAX/DAS structures can't work because there are no duplicated zeros to compress.
Therefore, we propose a new customized structure, namely RH, for Hadamard gate.
First, the special regularity of MTP is explored.
Furthermore, the MVM is conducted by parsing the RH on demand, which is optimized to $O(log_2N) = n$ for execution acceleration.


\subsubsection{RH structure}

With our observation, the MTP result of $n$ Hadamard gates possesses strong regularity, it produces a $2^n \times 2^n$ operation-matrix.
In this matrix, each element has the same absolute value, only the sign is different.
Moreover, the sign can be calculated based on the index of each element in operation-matrix $H^{\otimes n}$.
Therefore, there is no need to store the complete huge $2^n \times 2^n$ matrix.
We propose a new regularity-oriented structure, namely RH, for Hadamard gates.
The structure only stores the absolute value $\frac{1}{\sqrt{N}}$ for $n$-qubit Hadamard applying, where $N = 2^n$.
In this way, the memory usage is reduced from $2^n \times 2^n$ elements to only one element, which brings $2^{2n} \times$  memory saving.
When the $H^{\otimes n}$ matrix is required to participate in the MVM for next state-vector, the sign of each element is parsed.

\subsubsection{MVM acceleration with RH}

\begin{figure}[t]
\center
\includegraphics[width=0.9\linewidth]{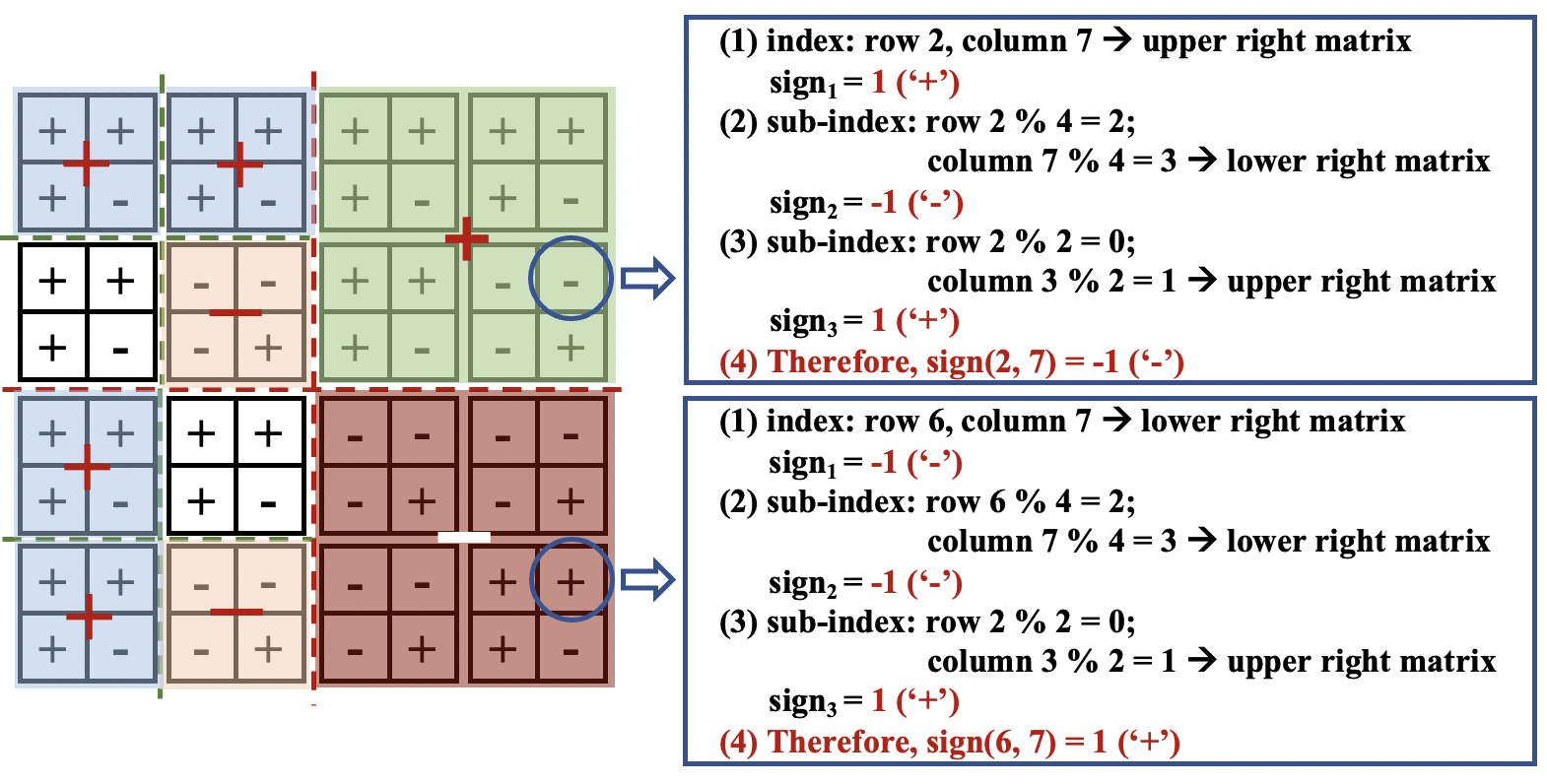}
\captionsetup{font={small}}
\setlength{\abovecaptionskip}{3pt}
\caption{The process of sign computation in matrix $H^{\otimes n}$ for a given index. Here, the number of qubits $n = 3$.}
\label{fig:h_sign}
\end{figure}

In this section, we propose three methods for MVM of $H^{\otimes n}$ and the state-vector $S$ for simulation acceleration.
When conducting MVM, each element of state-vector multiplies the same absolute value for each $H^{\otimes n}$ row.
The critical point is the sign of each element in $H^{\otimes n}$, which decides the final result of corresponding amplitude.
Given an index of an element in the $H^{\otimes n}$ matrix, we can compute the sign using $O(log_2N) = O(n)$ complexity.
We demonstrate the process of sign computation for matrix $H^{\otimes n}$ in Figure \ref{fig:h_sign}.

For matrix $H^{\otimes n}$, it can be divided into four equal sub-matrices, each of which owns an overall sign.
The overall sign of the upper left matrix, upper right matrix and lower left matrix are positive.
And the lower right matrix is negative.
Then we continue to divide a sub-matrix into four equal parts.
And the sign in the corresponding position is the same.
Given an index of an element, we can determine the relative position of the element in the sub-matrix layer by layer until the sub-matrix is no longer separable.
Then we multiply all the signs together to get the final sign of the element.
Here, we give two examples in Figure \ref{fig:h_sign}.

Given an index with row 2 and column 7, we can determine which sub-matrix the element is in.
Then the sub-sign is determined.
Finally, we compute the sign(2, 7) = -1.
Without optimization, for the MVM of $H^{\otimes n}$ and state-vector $S$, we first perform the scalar-multiply for the $S$ by $\frac{1}{\sqrt{N}}$.
In this way, the amplitudes don't need to multiply the same absolute value again and again.
Then, we calculate the sign of each element within the row.
Afterwards, the dot-product of this row and the state-vector is conducted to obtain an amplitude of the next state-vector.
For this non-optimized method, it requires $O(2^{2n})$ times sign computation.
In order to optimize the sign computation of $H^{\otimes n}$, we further propose three methods.

\begin{figure}[t]
\center
\includegraphics[width=0.9\linewidth]{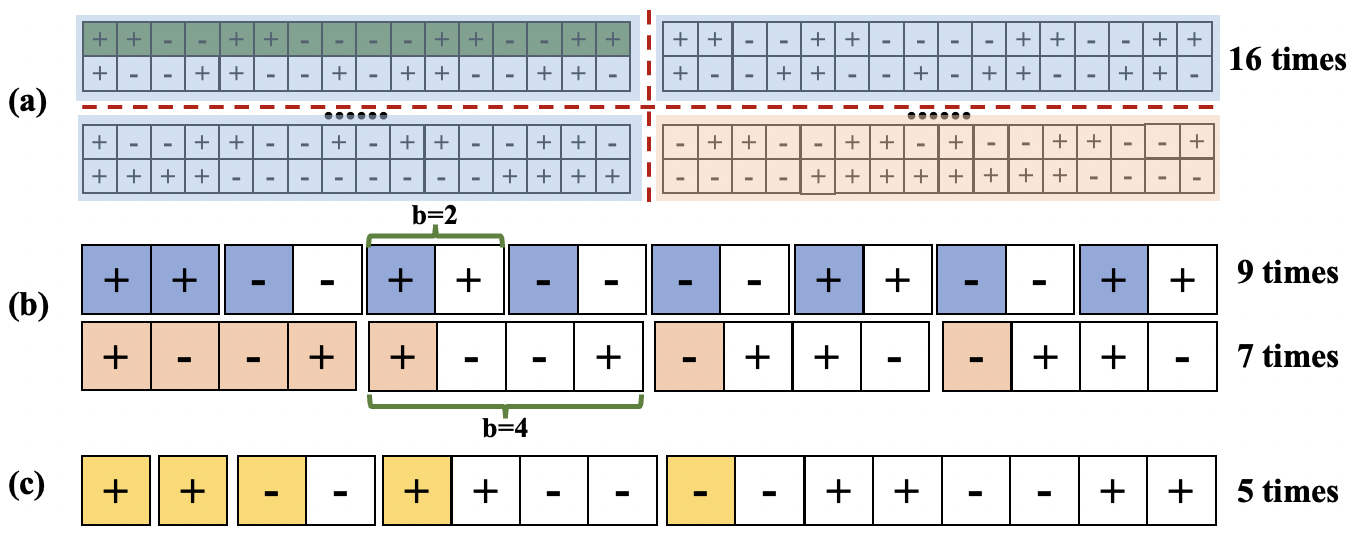}
\captionsetup{font={small}}
\setlength{\abovecaptionskip}{3pt}
\caption{Three sign computation methods of $H^{\otimes n}$. (a) quarter-based method; (b) block-based method; (c) logarithm-based method.
Here, the number of qubit $n$ = 5.}
\label{fig:mvm_sign}
\end{figure}

\textbf{\emph{Method 1: quarter-based method.}}
With our observation, the first half signs and the second half signs within a row of the operation-matrix $H^{\otimes n}$ have a certain correlation. 
For the rows from 0 to $2^{n-1} - 1$, the first half and second half signs in each row are the same.
And for the rows from $2^{n-1}$ to $2^n - 1$, the signs of two parts are opposite.
Fortunately, the first half signs from row 0 to $2^{n-1} - 1$ have the same sign distribution as the first half signs from row $2^{n-1}$ to $2^n - 1$.
Based on this regular pattern, only the quarter signs in the upper left sub-matrix need to be calculated.
Quarter-based method only requires $O(2^{2n-2})$ times sign computation, which brings 4$\times$ optimization than non-optimized method.
We give an example in Figure \ref{fig:mvm_sign} (a) to demonstrate the quarter-based method.
Only 16 rows of signs need to be computed, and for each row, the first 16 signs need to be computed.

\textbf{\emph{Method 2: block-based method.}}
Since the first half and the second half signs within each row have certain patterns, is there a similar rule for the division of smaller granularity?
With our observation, if the row is divided into blocks, the sign of each block is related to the first block.
Therefore, the second method for sign computation is named as block-based method.
Except the first block, the signs in a block are decided by the first sign of the block.
If the first sign in one block is positive, it means the signs in this block are the same as the signs of the first block.
On the contrary, if the first sign in this block is negative, which brings the opposite signs with the first block.

Therefore, in each row, all the elements in the first block need to be computed in advance.
Then in the extra blocks, only the first sign needs to be computed.
Similarly, based on the quarter-based method, we know that only quarter signs in the $2^n \times 2^n$ matrix need to be computed with $O(n)$ complexity.
If the number of signs in each block is $b$, then the number of sign computations can be calculated as $2^{n-1} * (\frac{2^{n-1}}{b} + b - 1)$.
If $b = 1$, the number of computations of block-based method is the same as the quarter-based method.
We give two examples in Figure \ref{fig:mvm_sign} (b) for the cases where $b = 2$ and $b = 4$.
The number of sign computation is 9 and 7, respectively.
Note that the number of elements in each block is the power of 2. Otherwise, the elements in a row cannot be divided equally.

\begin{theorem}\label{them:blk}
Let $b$ be the block size of block-based method for sign computation.
When $b = \sqrt{2^{n-1}}$, the method gets the minimum number of sign computation for $n$-qubit Hadamard gates applying.
\end{theorem}

\begin{proof}
As long as the number of sign computation in each row is minimal, thus the overall minimum number of sign computation can be obtained.
The number of sign computation in each row is $\frac{2^{n-1}}{b} + b - 1$.
We can formulate the number of sign computation as:
\begin{equation}
\begin{small}
    f(b) = \frac{2^{n-1}}{b} + b - 1,
\end{small}
\end{equation}
where $n$ is a constant.
In order to compute the $b$, which has the minimum value of $f(b)$, we calculate the derivation of the function as {\small $f'(b) = \frac{-2^{n-1}}{b^2} + 1$}.
When {\small $b = \sqrt{2^{n-1}}$}, {\small $f'(b) = 0$}.
Meanwhile, when {\small$b < \sqrt{2^{n-1}}$}, {\small$f'(b) < 0$}; when {\small$b > \sqrt{2^{n-1}}$}, {\small$f'(b) > 0$}.
Therefore, {\small$f(\sqrt{2^{n-1}})$} gets the minimum value.
\end{proof}

\begin{table}[t]
\captionsetup{font={small}}
\setlength{\abovecaptionskip}{5.3pt}
\caption{Summary of the number of sign computation for four methods.}
\def\arraystretch{1.3}\tabcolsep 18pt
\def\thefootnote{a}\footnotesize
\begin{center}
\begin{tabular}{|l|l|}
\hline
\textbf{Methods} & \textbf{Number of sign computation} \\
\hline
non-optimized & $2^n * 2^n$ \\
\hline
quarter-based  & $2^{n-1} * 2^{n-1}$ \\
\hline
block-based & $2 ^ {n-1} * (\frac{2^{n-1}}{b} + b - 1)$ \\
\hline
logarithm-based & $2 ^{n-1} * (n-1)$ \\
\hline
\end{tabular}%
\end{center}
\label{tab:sign_complx}%
\end{table}

\renewcommand{\algorithmicrequire}{\textbf{Input:}}
\renewcommand{\algorithmicensure}{\textbf{Output:}}
\begin{algorithm}[b]
\caption{RH-MVM based on logarithm-based method.}
\label{alg:mvm_sign}
\begin{algorithmic}[1]
\REQUIRE{the current state-vector $S$ of the $n$-qubit quantum circuit and the RH structure of $H^{\otimes n}$.}
\ENSURE{the next state-vector $S_{next}$ for the quantum circuit.}
\STATE scalar multiplication for RH.$value \cdot S$.
\STATE array sign[$2^{n-1}$] for temporarily storing the signs.
\FOR{$r = 0$ to $2^{n-1}-1$}
    \FOR{$c = 0$ to $2^{n-1}-1$}
        \IF{c == 0}
        \STATE sign[c] = 1;
        \ELSIF{c is the power of 2 ($2^k$)}
            \STATE{sign[c] = compute\_sign(r, c, n)} //O(n)
        \ELSE
            \STATE{sign[c] = sign[$\lfloor \frac{c}{2^{p-1}} \rfloor * 2^{p-1}$] * sign[c \% $2^{p-1}$]}
        \ENDIF
        \STATE{$S_{next}$[r] += sign[c] * S[c]}
        \STATE{$S_{next}$[r] += sign[c] * S[c+$2^{n-1}$]} 
       \STATE{$S_{next}$[r+$2^{n-1}$] += sign[c] * S[c]}
       \STATE{$S_{next}$[r+$2^{n-1}$] += -sign[c] * S[c+$2^{n-1}$]}
    \ENDFOR
\ENDFOR
\end{algorithmic}
\end{algorithm}

\textbf{\emph{Method 3: logarithm-based method.}} 
With further consideration, we can reduce the times of sign computation to $log_2{2^{n-1}} = n - 1$ for each row.
We name this method as logarithm-based method.
For each row of the matrix, we only compute the signs, whose indices are the power of the 2.
That's to say, except computing the first sign with index 0, the signs with index {\small $2^0, 2^1, 2^2, ....$} need to be computed.
When the sign with index $2^i$ is computed, the consecutive $2^i$ signs are computed.
If the sign of index $2^i$ is positive, then the following signs are the same as the signs that have been calculated before.
Conversely, if the sign is negative, the next signs are the opposite with the former signs.
Therefore, the number of computations can be calculated as $2^{n-1} * (n - 1)$.
We give an example in Figure \ref{fig:mvm_sign} (c), which only needs to compute 5 signs for each row.
The complexity of these three methods for sign computation compared with non-optimized method are listed in Table \ref{tab:sign_complx}.
Finally, we demonstrate the MVM of $H^{\otimes n}$ and $S$ using logarithm-based method in Algorithm \ref{alg:mvm_sign}.
Line 5-10 is for the sign computation of the first quarter, which is embedded into the RH-MVM algorithm.

\begin{property}
The complexity of RH-MVM algorithm is {\small $O(n \cdot 2^{2(n-1)})$}, where the embedded \emph{Method 3} of sign computation conduct a complexity of {\small $O(log_22^{n-1}) = O(n-1)$}.
\end{property}

\section{Experimental Results}\label{sec:case}

\subsection{Experimental Setup}
In this section, we evaluate our proposed DAX and DAS structures for sparse quantum gates and the RH structure for the dense Hardamard gate compared with the `Matrix' method.
We conduct our optimization methods on Intel(R) Xeon(R) Silver 4114 CPU @ 2.20GHz with 10 CPU cores. Our workstation possesses 0.7 TB memory.
Experiments are performed on GCC 9.3.0.

\subsection{Evaluation for DAX and DAS structures} \label{sec:exp_dax_das}

\begin{figure}[t]
\center
\includegraphics[width=0.9\linewidth]{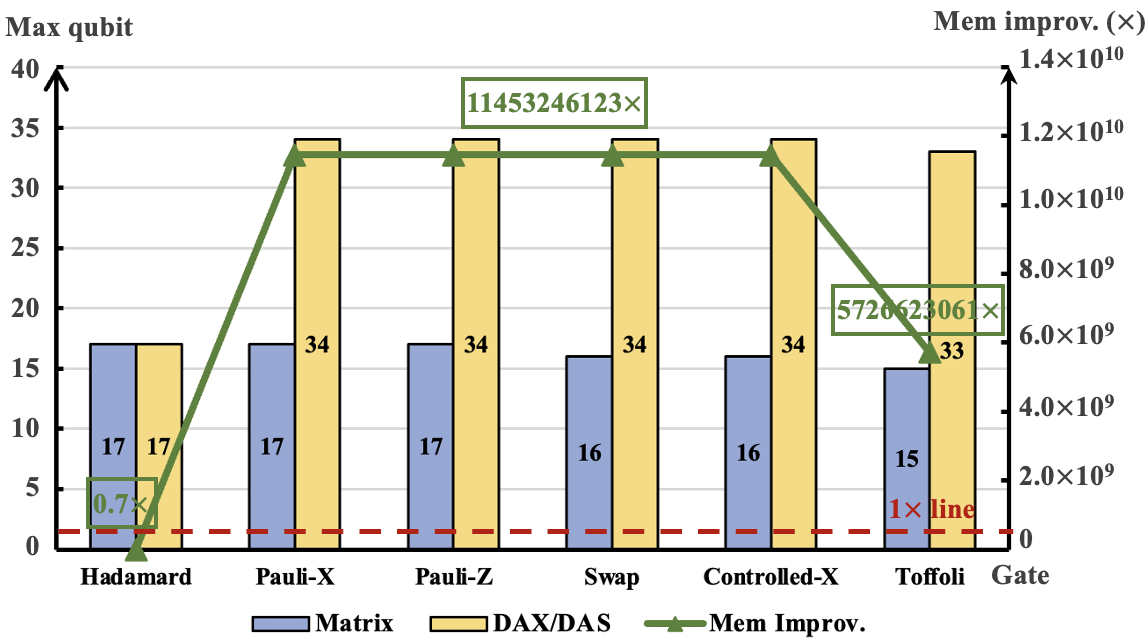}
\captionsetup{font={small}}
\setlength{\abovecaptionskip}{3pt}
\caption{Maximum number of qubits and the memory improvement of DAX and DAS compared with original `Matrix' structure.}
\label{exp:qDAXS}
\end{figure}

In this section, we demonstrate the efficiency of DAX and DAS.
We implement six types of quantum gates including dense Hadamard gate, single-qubit Paul-X and Pauli-Y gates, two-qubit Swap and Controlled-X gates, and the three-qubit Toffoli gate.
First, the maximum number of qubits for $m$ identical gates MTP is demonstrated.
Corresponding memory improvement for the maximum qubit is explored.
Then, the execution time for the MTP and MVM is counted to show the speedup of our method.

Figure \ref{exp:qDAXS} demonstrates the maximum number of qubits and the memory improvement under six types of quantum gates.
We conduct the MTP using $m$ identical quantum gates to approach the maximum number of qubits.
We compare the original `Matrix' structure without any compression with our DAX and DAS.
For the `Matrix' structure, our workstation can accommodate up to 17 qubits for single-qubit gates.
And for the two-qubit and three-qubit gates, $n$ needs to be the multiple of 2 and 3.
Therefore, the maximum number of qubits is 16 and 15 for two-qubit and three-qubit quantum gates, respectively.
Obviously, using our DAX and DAS structures, the maximum number of qubits increases to 34.
And the DAX and DAS structures achieve $1.1 \times 10^{10} \times$ memory saving compared with the original structure.
However, for the dense Hadamard gate, the memory usage can't be saved, which leads that the maximum qubit that can be accommodated is still 17.

\begin{figure}[t]
\center
\includegraphics[width=0.9\linewidth]{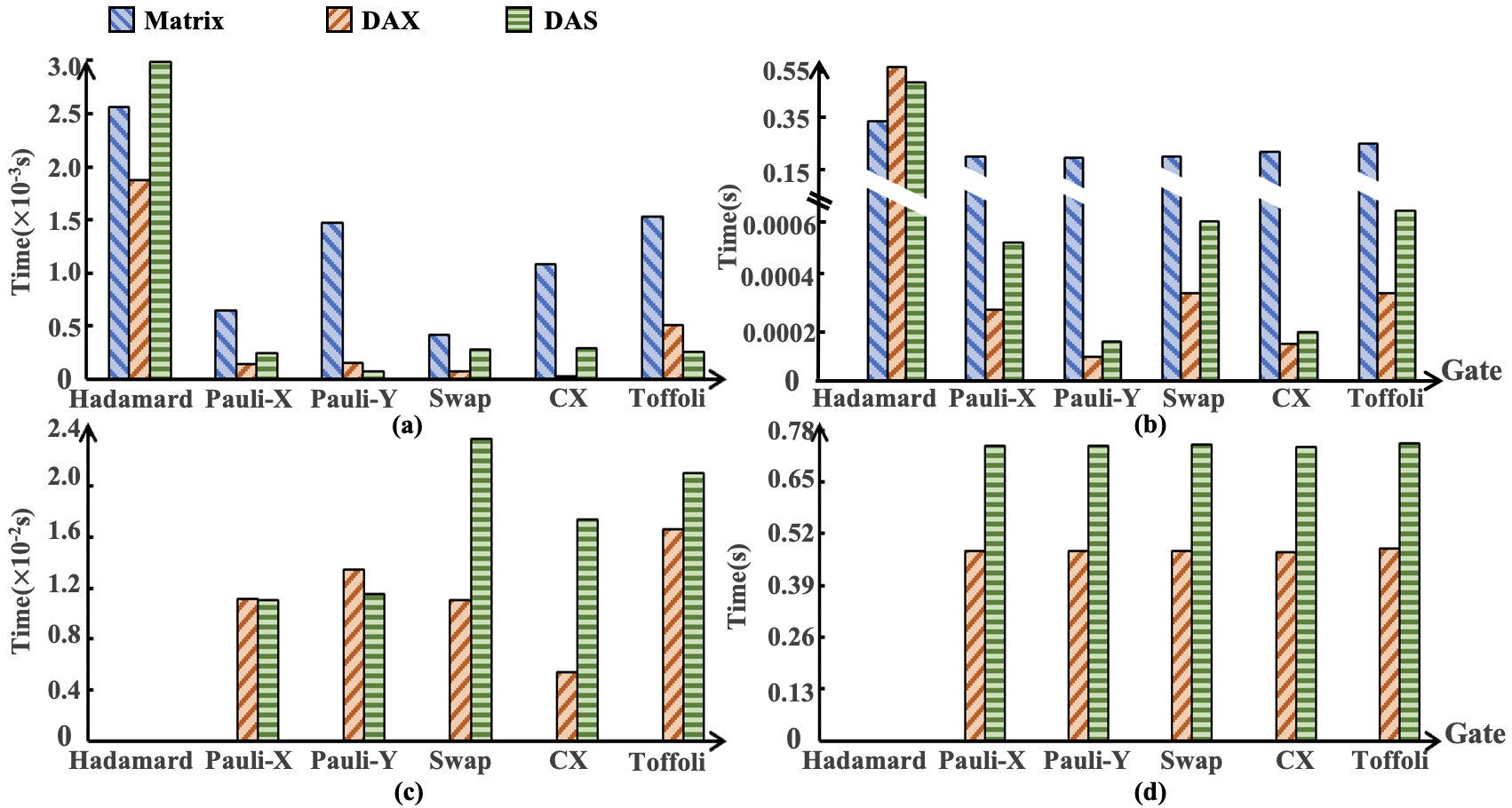}
\captionsetup{font={small}}
\setlength{\abovecaptionskip}{3pt}
\caption{The execution time comparison for MTP computation of Matrix, DAX and DAS. (a) 6-qubit circuit; (b) 12-qubit circuit; (c) 18-qubit circuit; (d) 24-qubit circuit.}
\label{exp:ti_MTP_DAXS}
\end{figure}

\begin{figure}[t]
\center
\includegraphics[width=0.9\linewidth]{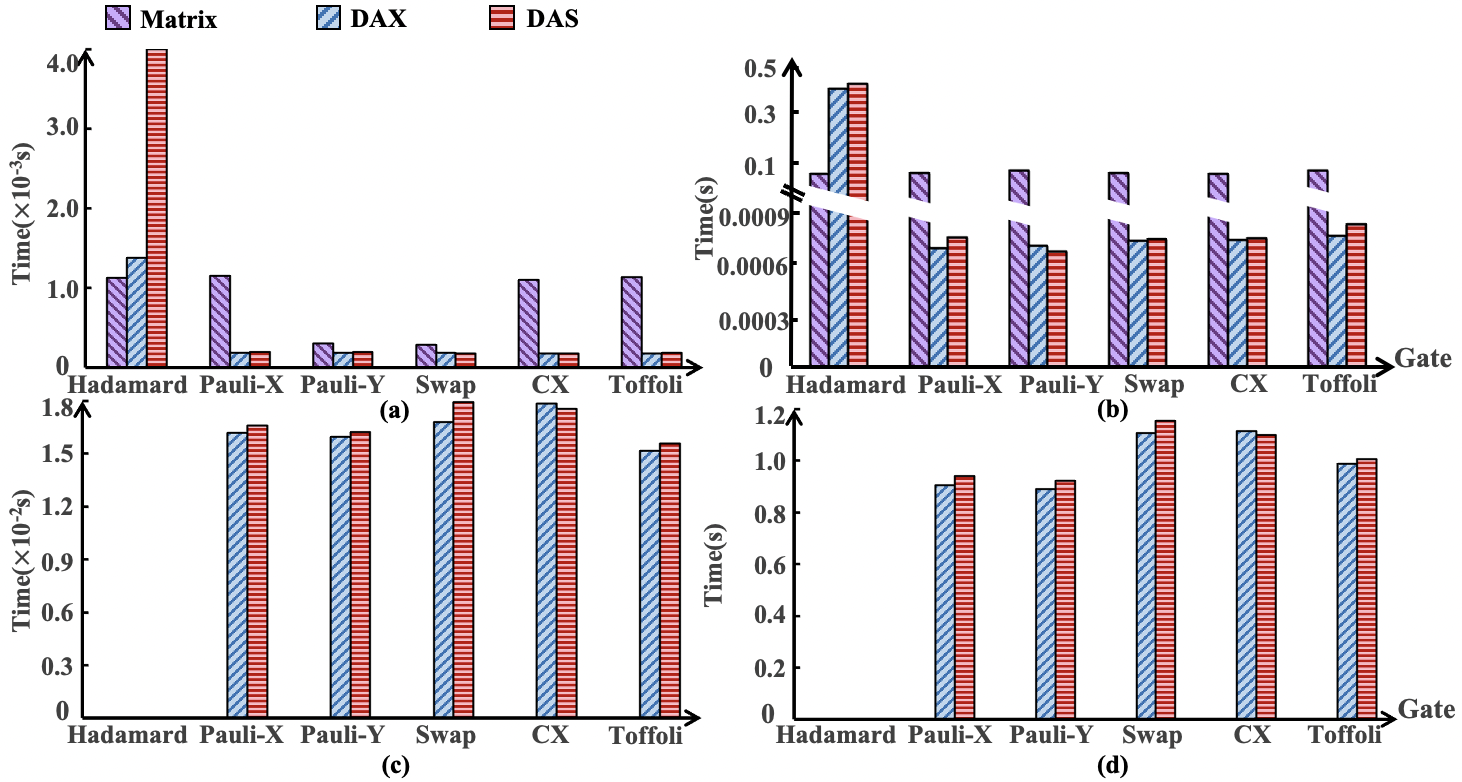}
\captionsetup{font={small}}
\setlength{\abovecaptionskip}{3pt}
\caption{The execution time comparison for MVM computation of Matrix, DAX and DAS. (a) 6-qubit circuit; (b) 12-qubit circuit; (c) 18-qubit circuit; (d) 24-qubit circuit.}
\label{exp:ti_MVM_DAXS}
\end{figure}

Figure \ref{exp:ti_MTP_DAXS} and \ref{exp:ti_MVM_DAXS} demonstrate the execution time comparison for MTP and MVM computations.
We select four quantum circuits with different qubits (i.e., 6, 12, 18, 24) for comparison.
The circuits are also implemented by $m$ identical quantum gates.
For Hadamard gate, because the non-zero elements in DAX and DAS can't be decreased, the execution time can't be accelerated.
And because the DAX and DAS need to parse the index for computation, they need more time to simulate.
Moreover, for the 18-qubit and 24-qubit simulation, the memory can't afford the memory usage.
However, for the other five types of sparse gates, the DAX and DAS can achieve up to 1800.3$\times$ execution speedup for 12-qubit MTP computation.
And for MVM computation, DAX and DAS can achieve up to 141.2$\times$ speedup for 12-qubit quantum circuit.
Obviously, as the qubit increases, DAX structure shows better simulation acceleration than DAS.

\subsection{Evaluation for RH structures}

In this section, we implement the RH of $H^{\otimes n}$ and the MVM of $H^{\otimes n}$ and state-vector $S$ on our workstation with 0.7 TB of memory.
For the MTP of $n$ Hadamard gates, the memory can accommodate any qubits using only 3 microseconds to execute.
For MVM of the $H^{\otimes n}$ and $S$, we implement four sign computation methods, including non-optimized, quarter-based, block-based and the logarithm-based methods.
Correspondingly, the methods coupled with RH structure are marked as `RH\_Noptimized', `RH\_Quarter', `RH\_Block' and `RH\_Logarithm' in Figure \ref{exp:RH}.
For the block-based method, the block size $b$ is set as {\small $\sqrt{2^{n-1}}$}, which is proved in the Theorem \ref{them:blk}.
Meanwhile, the results of the method without RH optimization, marked as `Matrix', are demonstrated for comparison.

For the MVM computation, our workstation can accommodate up to 35 qubits using RH structure.
But the maximum number of qubit is only 17 using original matrix representation. 
Our RH structure achieves over $2^n \times$ memory saving. 
As for the simulation time, we implement above five methods for comparison. Results are demonstrated in Figure \ref{exp:RH}.
The simulation time is demonstrated from qubit 12 to 17.
When the qubit is 18, `RH\_Noptimized' will approach 20 hours for simulation.
We don't involve so long time for comparison, because the time can be estimated with the linear growth.
Results show that our proposed `RH\_Logarithm' method achieves the fastest quantum circuit simulation.
For qubit 17, `RH\_Noptimized' method consumes 4.47 hours for simulation.
However, our `RH\_Logarithm' method only needs 163.3 seconds, which achieves 98.7$\times$ speedup.
Meanwhile, for the original `Matrix' method, our `RH\_Logarithm' method also produces 2.3$\times$ acceleration.

\begin{figure}[t]
\center
\includegraphics[width=1.0\linewidth]{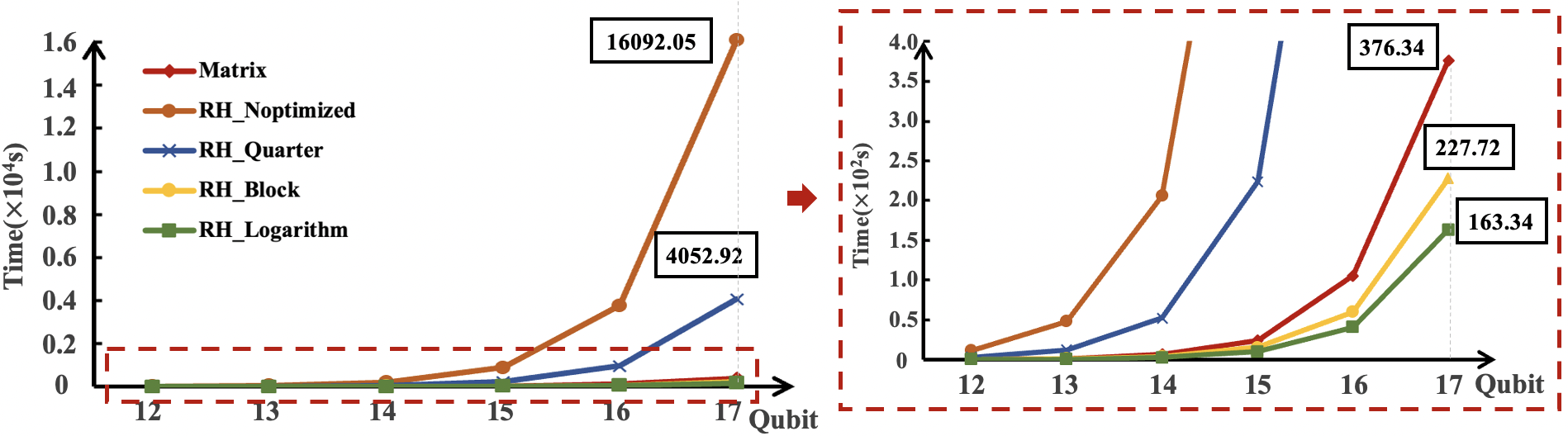}
\captionsetup{font={small}}
\setlength{\abovecaptionskip}{3pt}
\caption{Execution time comparison of MVM computation for dense Hadamard gates.}
\label{exp:RH}
\end{figure}

\begin{table*}[t]
\captionsetup{font={small}}
\setlength{\abovecaptionskip}{3pt}
\caption{Methods comparison for a QNN neuron.}
\def\arraystretch{1.3}\tabcolsep 14pt
\def\thefootnote{a}\footnotesize
\begin{center}
\begin{tabular}{|c|c|ccc|cccc|}
\hline
\multirow{2}[2]{*}{\textbf{Methods}} & \multirow{2}[2]{*}{\textbf{Max Qubit}} & \multicolumn{3}{c|}{\textbf{Memory Usage}} & \multicolumn{4}{c|}{\textbf{Execution Time (s)}} \\
      &       & stage1 & stage2 & stage3 & 4     & 7     & 12    & 21 \\
\hline
Matrix & 12    & 256MB & 256MB & 256MB & 0.000662  & 0.002292  & 2.045140  & - \\
\hline
\multirow{2}[2]{*}{RH-DAX} & \multirow{2}[2]{*}{12} & 256MB & 160KB & 160KB & 0.000363  & 0.001396  & 0.295478  & - \\
      &       & 1.0$\times$ & 1638.4$\times$ & 1638.4$\times$ & 1.8$\times$ & 1.6$\times$ & 6.9$\times$ & - \\
\hline
\end{tabular}%
\end{center}
\label{tab:exp_qnn}%
\end{table*}

\begin{table*}[t]
\captionsetup{font={small}}
\setlength{\abovecaptionskip}{3pt}
\caption{Methods comparison for Grover algorithm.}
\def\arraystretch{1.3}\tabcolsep 8.53pt
\def\thefootnote{a}\footnotesize
\begin{center}
\begin{tabular}{|c|c|ccc|cccccc|}
\hline
\multirow{2}[2]{*}{\textbf{Methods}} & \multirow{2}[2]{*}{\textbf{Max Qubits}} & \multicolumn{3}{c|}{\textbf{Memory Usage}} & \multicolumn{6}{c|}{\textbf{Execution Time}} \\
      &       & $H^{\otimes n}$ & Oracle & $Z_0$  & 3     & 6     & 9     & 12    & 15    & 17 \\
\hline
Matrix & 17    & 256GB & 256GB & 256GB & 0.000177  & 0.008643  & 0.398444  & 73.15  & 12310.37  & 347099.60  \\
\hline
\multirow{2}[2]{*}{RH-DAX} & \multirow{2}[2]{*}{34} & \multirow{2}[2]{*}{256GB} & \multirow{2}[2]{*}{0.6TB} & \multirow{2}[2]{*}{0.6TB} & 0.000044  & 0.002235  & 0.088259  & 16.04  & 2907.44  & 92256.08  \\
      &       &       &       &       & 4.0 $\times$ & 3.9 $\times$ & 4.5 $\times$ & 4.6 $\times$ & 4.2 $\times$ & 3.8 $\times$ \\
\hline
\end{tabular}%
\end{center}
\label{tab:exp_grover}%
\end{table*}

\subsection{Evaluation on QNN neuron and Grover algorithm}

In this section, we implement two quantum algorithms, a QNN neuron from work \cite{jiang2021co} and the Grover algorithm, for a case study to show the efficiency of our proposed methods.
Figure \ref{fig:qnn} shows the quantum circuit of the QNN neuron, which includes data encoding (stage 1), dot-product of inputs and binary weight (stage 2), encoding the QNN results to qubit $O$ (stage 3), and measurement (stage 4).
The proposed compressed structure for sparse quantum gates and customized RH structure for dense Hadamard gates are implemented for comparison. 
Results of the maximum number of qubits, memory usage, execution time and corresponding improvement are demonstrated in Table \ref{tab:exp_qnn} and \ref{tab:exp_grover}.
For the `Matrix' method, all the data are represented using original matrix without any optimizations.
The duplicated zeros and regular Hadamard gates are not compressed.
Then, for `RH-DAX' method, because two algorithms rely on various types of quantum gates, including the sparse gates and dense Hadamard gate, corresponding optimizations are implemented.
For Hadamard gate, the customized RH is implemented using logarithm-based MVM method.
For the sparse gates, the DAX is utilized.
Because, in previous experiments, the DAX structure demonstrates better performance than DAS.

\begin{figure}[t]
\center
\includegraphics[width=0.9\linewidth]{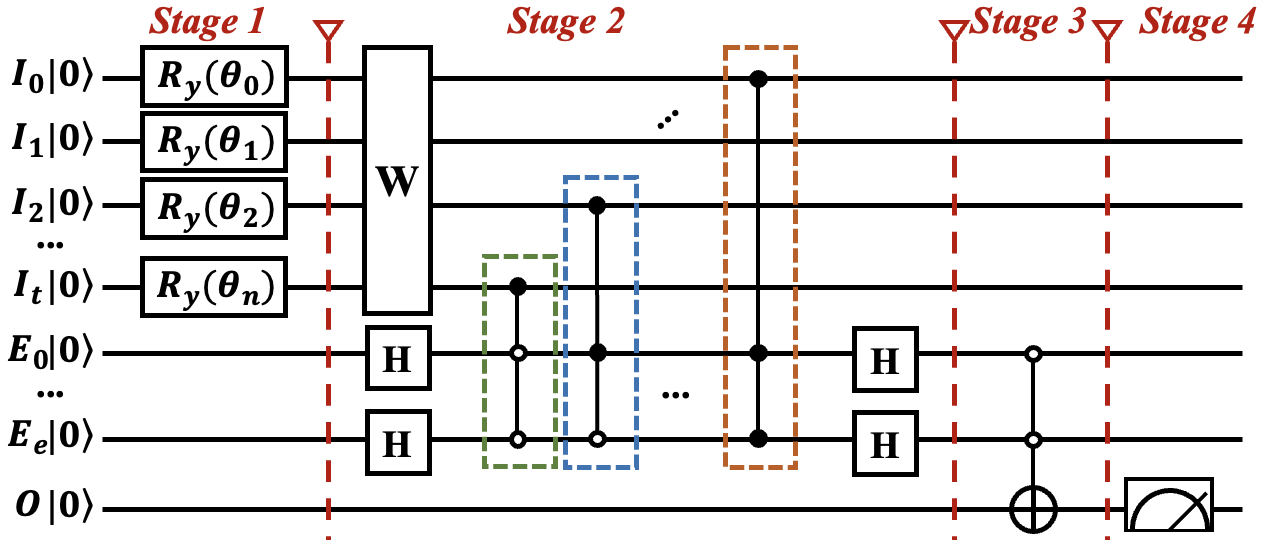}
\captionsetup{font={small}}
\setlength{\abovecaptionskip}{3pt}
\caption{The quantum circuit of a QNN neuron \cite{jiang2021co}.}
\label{fig:qnn}
\end{figure}

The results of different methods for QNN neuron are demonstrated in Table \ref{tab:exp_qnn}.
For the `Matrix' method without optimizations, our workstation can only simulate 12 qubits, which includes 8 inputs-qubits, 3 encode-qubits and 1 output-qubit.
For each stage, the memory needs to accommodate $2^n \times 2^n$ operation-matrix and $2^n$ state-vector.
Therefore, they all require about 256 MB memory.
However, the 256 MB is far from the 0.7 TB memory of our workstation.
This is because the number of inputs for a quantum neuron needs to be the power of 2.
For 16 inputs, it requires 21 qubits for simulation, the memory requirement far exceeds the memory boundary.

For the `RH-DAX' method, the specific Hadamard gates are implemented using our proposed RH structure.
And for other sparse gates, we utilize the DAX structure for data compression.
Because the data encoding stage of the neuron applies dense RY gates, which can't save the memory, it also needs 256 MB.
However, for stage 2 and stage 3, it requires only 160 KB memory because of the compressed DAX.
As we can see, even if the hybrid RH and DAX method achieves up to 1638.4$\times$ memory saving for stage 2 and stage 3, the maximum number of qubits for this method is still 12.
Because the dense RY gates in stage 1 become the bottleneck for qubits increasing.
This problem will be considered in future work.
Finally, we discuss the execution time for different methods under different numbers of qubits.
The speedup listed in the table is compared with `Matrix' method.
The `RH-DAX' method achieves up to 6.9$\times$ acceleration compared with the `Matrix' method, when the number of qubits is 12.


For Grover algorithm, we compare our proposed hybrid `RH-DAX' method with the original representation using a matrix structure.
Similarly, the maximum number of qubits, memory usage, execution time and corresponding improvement compared with `Matrix' method are demonstrated in Table \ref{tab:exp_grover}.
Specially, Grover algorithm supports consecutive qubits for simulation.
An $n$-qubit quantum circuit means the database possesses $N = 2^n$ elements for searching.
For original `Matrix' method, our workstation can accommodate 17 qubits.
However, our hybrid `RH-DAX' method can achieve up to 34-qubit simulation.
Based on the type of quantum gates, we classify the circuit into three parts: $H^{\otimes n}$, oracle and $Z_0$ transformation.
For each part, we list the maximum memory usage for the maximum qubits.
For the execution time, we list six types of qubits up to 17 qubits for comparison, because the `Matrix' method can only accommodate 17 qubits.
Our hybrid `RH-DAX' method achieves up to 4.6$\times$ simulation speedup than `Matrix' method. For these six qubits, the speedup is 4.2$\times$ on average, which is a stable acceleration.

\section{Conclusion}\label{sec:conclusion}

In order to accommodate more qubits and accelerate quantum circuit simulation on personal computing capability, this paper proposes two-level optimizations via redundancy exploration.
First, based on a large number of sparse quantum gates, we propose two compressed structures for duplicated zeros compression.
The design doesn't require repeat compression and decompression, it can seamlessly update the state-vector.
This method increases the number of qubits from 17 to 34 and achieves up to 1800.3$\times$ and 141.2$\times$ acceleration for MTP and MVM operations.
Second, for the dense and definitely used Hadamard gate, we design an RH structure by exploring the regularity of Hadamard applying.
The structure is parsed on demand, and the state-vector updating process is accelerated by 2.3$\times$.
Finally, experiments are conducted on actual quantum algorithms. The mixed DAX and RH method achieves up to 6.9$\times$ and 4.6$\times$ speedup for the QNN and Grover algorithm, respectively.

\bibliographystyle{ieeetr}
\bibliography{ref}

\end{document}